\documentclass[notitlepage,a4,10.5pt]{article}

\usepackage{graphicx}

\usepackage{amsmath}%
\usepackage{amsfonts}%
\usepackage{amssymb}%
\usepackage{mathrsfs}
\usepackage{amsthm}
\usepackage{float}

\usepackage{authblk}

\usepackage[left=2.5cm,right=2.5cm,top=3cm,bottom=3cm]{geometry} 

\usepackage{xcolor}

\newcommand{\mc}{\mathcal}
\newcommand{\cp}{\times}

\newcommand{\Norm}[1]{\left\lvert\left\lvert #1 \right\rvert\right\rvert}

\newcommand{\bol}{\boldsymbol}

\newcommand{\abs}[1]{\left\lvert{#1}\right\rvert}

\newcommand{\lr}[1]{\left({#1}\right)}
\newcommand{\lrs}[1]{\left[{#1}\right]}
\newcommand{\lrc}[1]{\left\{{#1}\right\}}

\newcommand{\mf}{\mathfrak}
\newcommand{\p}{\partial}

\newcommand{\ti}[1]{\textit{#1}}
\newcommand{\tb}[1]{\textbf{#1}}

\newtheorem{remark}{\textit{Remark}}
\newtheorem{theorem}{\textit{Theorem}}

\begin{document}

\title{
Hamiltonian Structure and Nonlinear Stability of Steady   
Solutions \\of the Generalized Hasegawa-Mima Equation \\for Drift Wave   
Turbulence in Curved Magnetic Fields}
\author[1]{Naoki Sato} 
\author[2]{Michio Yamada}
\affil[1]{Graduate School of Frontier Sciences, \protect\\ The University of Tokyo, Kashiwa, Chiba 277-8561, Japan \protect\\ Email: sato\_naoki@edu.k.u-tokyo.ac.jp}
\affil[2]{Research Institute for Mathematical Sciences, \protect\\ Kyoto University, Kyoto 606-8502, Japan
\protect \\ Email: yamada@kurims.kyoto-u.ac.jp}
\date{\today}
\setcounter{Maxaffil}{0}
\renewcommand\Affilfont{\itshape\small}

    \maketitle
    \begin{abstract}
    The Generalized Hasegawa-Mima (GHM) equation, which generalizes the standard Hasegawa-Mima (HM) equation, is a nonlinear equation 
    describing the evolution of drift wave turbulence in curved magnetic fields.  
    The GHM equation can be obtained from a drift wave turbulence ordering 
    that does not involve ordering conditions on spatial derivatives of the magnetic field or the plasma density, and it is therefore appropriate to describe the evolution of electrostatic turbulence in strongly inhomogeneous magnetized plasmas. 
    In this work, 
   we discuss the noncanonical Hamiltonian structure of the GHM equation,  
   and obtain conditions for the nonlinear stability of steady solutions through the energy-Casimir stability criterion. 
   These results are then applied to describe drift waves and infer the existence of stable toroidal zonal flows with radial shear in dipole magnetic fields.
    \end{abstract}

\section{Introduction}
The generalized Hasegawa-Mima (GHM) equation \cite{GHM,GHM2}
\begin{equation}
\frac{\p}{\p t}\lrs{\lambda A_e\chi-\sigma\nabla\cdot\lr{A_e\frac{\nabla_{\perp}\chi}{B^2}}}=\nabla\cdot\lrs{A_e\lr{\sigma\frac{\bol{B}\cdot\nabla\cp\bol{v}_{\bol{E}}^{\chi}}{B^2}-1}\bol{v}_{\bol{E}}^{\chi}},\label{GHMgc5}
\end{equation}
describes the nonlinear evolution of the field $\chi\lr{\bol{x},t}=\varphi\lr{\bol{x},t}+\frac{\sigma}{2}\bol{v}_{\bol{E}}^2\lr{\bol{x},t}$,
physically representing the energy of a charged particle, caused by drift wave turbulence in an ion-electron plasma within a static magnetic field $\bol{B}\lr{\bol{x}}\neq\bol{0}$ of arbitrary geometry.  
Here $\bol{x}$ are Cartesian coordinates in a region $\Omega\subseteq\mathbb{R}^3$, $t$ is the time variable, $A_e\lr{\bol{x}}$ 
the leading order electron spatial density, 
$\sigma=m/Ze$ a physical constant with $m$ and $Ze$ ion mass and charge, $Z\in\mathbb{N}$, $\nabla_{\perp}=-B^{-2}\bol{B}\cp\lr{\bol{B}\cp\nabla}$, 
and the velocity fields $\bol{v}_{\bol{E}}^{\chi}\lr{\bol{x},t}$ and $\bol{v}_{\bol{E}}\lr{\bol{x},t}$ are respectively defined as
\begin{equation}
\bol{v}_{\bol{E}}^{\chi}=\frac{\bol{B}\cp\nabla\chi}{B^2},~~~~\bol{v}_{\bol{E}}=\frac{\bol{B}\cp\nabla\varphi}{B^2},
\end{equation}
with $\bol{E}\lr{\bol{x},t}=-\nabla\varphi\lr{\bol{x},t}$ the electric field associated with the electrostatic potential $\varphi\lr{\bol{x},t}$. The plasma is quasineutral, implying that $n_e=A_ee^{\lambda\varphi}=Z n_i$ with $n_e$ and $n_i$ the electron and ion densities, and  $\lambda=e/k_BT_e$ a physical constant 
where $k_BT_e$ denotes the temperature of the thermalized electron component. 

\begin{table}
\begin{center}
\begin{tabular}{c c c c c c} 
 \hline
 \hline
 Invariant & Expression & Field Conditions & Boundary Conditions\\  
 \hline 
 Mass $M_{\Omega}$ & $\int_{\Omega}A_e\lr{1+\lambda\chi}\,d\bol{x}$ & none &  $A_e{\bol{V}}_{dw}\cdot\bol{n}=0$ &\\
 Energy $H_{\Omega}$ & $\frac{1}{2}\int_{\Omega}A_e\lr{\lambda\chi^2+\sigma\frac{\abs{\nabla_{\perp}\chi}^2}{B^2}}\,d\bol{x}$ & none & $A_e\chi{\bol{V}}_{dw}\cdot\bol{n}=0$&\\
Enstrophy $W_{\Omega}$ & $\int_{\Omega}A_ew\lr{\lambda\chi-\frac{\sigma}{A_e}\nabla\cdot\lr{A_e\frac{\nabla_{\perp}\chi}{B^2}}}\,d\bol{x}$ & $\nabla\cp\lr{A_e\frac{\bol{B}}{B^2}}=\bol{0}$ & $wA_e\bol{v}_{\bol{E}}^{\chi}\cdot\bol{n}=0$ &\\
 \hline
 \hline
\end{tabular}
\caption{\label{tab3} Invariants of the GHM equation \eqref{GHMgc5}.}
\end{center}
\end{table}

Under suitable boundary conditions, the GHM equation \eqref{GHMgc5} preserves total ion mass $M_{\Omega}$ and energy $H_{\Omega}$, 
\begin{equation}
M_{\Omega}=\frac{m}{Z}\int_{\Omega}A_e\lr{1+\lambda\chi}d\bol{x},~~~~
H_{\Omega}=\frac{1}{2}\int_{\Omega}A_e\lr{\lambda\chi^2+\sigma\frac{\abs{\nabla_{\perp}\chi}^2}{B^2}}\,d\bol{x}.\label{MH} 
\end{equation}
In addition, a third invariant (generalized enstrophy) $W_{\Omega}$ arises when the magnetic field $\bol{B}$ and the electron spatial density $A_e$ satisfy the integrability condition $\nabla\cp\lr{A_e\bol{B}/B^2}=\bol{0}$,
\begin{equation}
W_{\Omega}=\int_{\Omega}A_ew\lr{\lambda\chi-\sigma\frac{\omega}{A_e}}\,d\bol{x},\label{W}
\end{equation}
where $w\lr{\lambda\chi-\sigma\frac{\omega}{A_e}}$ 
is any function of $\lambda\chi-\sigma\frac{\omega}{A_e}$ with $\omega=\nabla\cdot\lr{A_e\frac{\nabla_{\perp}{\chi}}{B^2}}$.
The invariants of the GHM equation are summarized in table 1. In this table $\bol{n}$ denotes the unit outward normal to the bounding surface $\p\Omega$, while the vector field (total drift velocity)  ${\bol{V}}_{dw}$ is defined as
\begin{equation}
{\bol{V}}_{dw}={\lr{1-\sigma\frac{\bol{B}\cdot\nabla\cp\bol{v}_{\bol{E}}^{\chi}}{B^2}}\bol{v}_{\bol{E}}^{\chi}}-\sigma{\frac{\nabla_{\perp}\chi_t}{B^2}},
\end{equation}
so that the GHM equation \eqref{GHMgc5} can be equivalently written as
\begin{equation}
\lambda A_e\frac{\p\chi}{\p t}=-\nabla\cdot\lr{A_e{\bol{V}}_{dw}}.
\end{equation}

The GHM equation \eqref{GHMgc5} can be derived 
from guiding center theory \cite{GHM2}    
by expanding the Euler-Lagrange equations arising from the Northrop guiding center Lagrangian \cite{Cary, Northrop} under an appropriate guiding center drift wave turbulence ordering, or from a two-fluid model of the ion-electron plasma under an equivalent two-fluid drift wave turbulence ordering \cite{GHM,Haz98}. 
Table 1 summarizes the guiding
center ordering required for the conservation of the first adiabatic invariant $\mu$ in guiding center theory, 
table 2 gives the drift wave turbulence ordering required to obtain the GHM equation \eqref{GHMgc5} from guiding center theory, while table 3 shows the drift wave turbulence ordering leading to \eqref{GHMgc5} from two-fluid theory.  
Here, $\epsilon >0$ denotes a small ordering parameter, $\omega_c=ZeB/m$ the ion cyclotron frequency, $\tau$, $\tau_d$, $\tau_b$ a reference time scale, the drift wave turbulence time scale, and the time scale of  bounce motion, $\rho$ the gyroradius, $L$ a characteristic scale length for the system, $\bol{E}_{\perp}$ the component of $\bol{E}$ perpendicular to $\bol{B}$, $\bol{E}_{\parallel}=\bol{E}-\bol{E}_{\perp}$, $\bol{v}$ the velocity of a charged particle, $\bol{v}_{\nabla}$, $\bol{v}_{\kappa}$, $\bol{v}_{\rm pol}$ the $\nabla B$, curvature, and polarization drifts, $u$ the guiding center velocity along $\bol{B}$, $k_BT_c$ the average energy of ion cyclotron motion, $E_{\parallel}'$ the component of the effective guiding center electric field $\bol{E}'$ along the effective guiding center magnetic field $\bol{B}'$ 
(see \cite{GHM2,Cary} for definitions),  $v_{\parallel}$ the ion fluid velocity parallel to $\bol{B}$, and $P$ the ion fluid pressure. 
In essence, the drift wave turbulence orderings of tables 2 and 3 describe an ion-electron plasma with  cold ions and a hot electron component where the dynamics along the magnetic field is slow compared to perpendicular $\bol{E}\cp\bol{B}$ drift motion. 

When $\bol{B}=B_0\nabla z$, $\log A_{e}=\log A_{e0}+\beta x$, $B_0,A_{e0},\beta\in\mathbb{R}$, $\beta L\sim\epsilon$, the GHM equation \eqref{GHMgc5} reduces to the standard Hasegawa-Mima (HM) equation \cite{HM,HM2}  
\begin{equation}
\frac{\p}{\p t}\lr{\lambda\varphi-\frac{\sigma}{B_0^2}\Delta_{\lr{x,y}}\varphi}=\frac{\sigma}{B_0^3}\lrs{\varphi,\Delta_{\lr{x,y}}\varphi}_{\lr{x,y}}+\frac{\beta}{B_0}\varphi_y.\label{HM}
\end{equation}
In this notation, $\lrs{f,g}_{\lr{x,y}}=f_xg_y-f_yg_x$,  $\Delta_{\lr{x,y}}=\p_{x}^2+\p_{y}^2$, and lower indexes denote partial derivatives, for example $f_x=\p f/\p x$.  
The HM equation represents 
a simple but effective model of 2-dimensional turbulence in magnetized plasmas and fluids \cite{Horton,Batchelor}, 
which exhibits self-organizing behavior (zonal flows)  \cite{HM5,Hasegawa85,Horton2,Fujisawa,Diamond} associated with 
inverse energy cascades \cite{Kraichnan2,Kraichnan,Rivera,Xiao}.
The nonlinearity of the equation is driven by the polarization drift, 
and in the presence of a density gradient the evolution of the electrostatic potential can be understood in terms of the nonlinear interaction of drift waves.
The geophysical fluid dynamics equivalent of the HM equation, the Charney equation \cite{Charney,Charney3} plays a central role in the understanding of atmospheric dynamics on the surface of rotating planets, 
with the Rossby wave replacing the drift wave of the HM system. 

\begin{table}
\begin{center}
\begin{tabular}{c c c c c c} 
 \hline
 \hline
 Order & Dimensionless & Fields & Distances & Rates & Velocities\\  
 \hline 
 $\epsilon^{-1}$ & & $\bol{B},\bol{E}_{\perp}$& & $\omega_{c}$ &\\ 
 $1$ &  & $\bol{E}_{\parallel}$ & $L$ & $\bol{v}/L,\bol{v}_{\bol{E}}/L,\tau^{-1}$ &$\bol{v},\bol{v}_{\bol{E}}$\\
 $\epsilon$ & $\rho/L$, $\lr{\omega_{c}\tau}^{-1}$ &  & $\rho$ & $\bol{v}_{\nabla}/L,\bol{v}_{\kappa}/L,\bol{v}_{\rm pol}/L$  & $\bol{v}_{\nabla},\bol{v}_{\kappa},\bol{v}_{\rm pol}$\\
 \hline
 \hline
\end{tabular}
\caption{\label{tab2} Guiding center ordering required for the existence of the first adiabatic invariant $\mu$ \cite{Cary}.}
\end{center}

\begin{center}
\begin{tabular}{c c c c c c} 
 \hline
 \hline
 Order & Dimensionless & Fields & Distances & Rates & Velocities\\  
 \hline 
 $\epsilon^{-1}$ & & $\bol{B},\bol{E}_{\perp}$& & $\omega_c$ &\\ 
 $1$ &  & $A_e$ & $L$ & $\tau^{-1}_d,\bol{v}_{\bol{E}}/L$ & $\bol{v}_{\bol{E}}$\\
 $\epsilon$ &  $\lambda\varphi,\rho/L,\lr{\omega_c\tau_d}^{-1}, k_BT_c/\frac{m}{2}\bol{v}_{\bol{E}}^2$ &  & $\rho$ & $\bol{v}_{\rm pol}/L$ & $\bol{v}_{\rm pol}$\\
 $\epsilon^2$ &$\frac{m}{2}\bol{v}_{\bol{E}}^2/k_BT_e,\tau_d/\tau_b$ & $E_{\parallel}'$ &  & $\bol{v}_{\nabla}/L,u/L,\tau_b^{-1}$ & $\bol{v}_{\nabla},u$ \\
 $\epsilon^5$ & & & & $\bol{v}_{\kappa}/L$ & $\bol{v}_{\kappa}$ \\
 \hline
 \hline
\end{tabular}
\caption{\label{tab3} Drift wave turbulence ordering for the derivation of the GHM equation in guiding-center theory \cite{GHM2}.}
\end{center}
\begin{center}
\begin{tabular}{c c c c c c} 
 \hline
 \hline
 Order & Dimensionless & Fields & Distances & Rates & Velocities\\  
 \hline 
 $1$ & & $\bol{B},A_e$& $L$ & $\omega_c$ &\\ 
 $\epsilon$ & $\lambda\varphi, 
 \omega_c^{-1}\p_t$ & $\bol{E}_{\perp}$ & & $\tau^{-1}_d,\bol{v}_{\bol{E}}/L$ & $\bol{v}_{\bol{E}}$\\
 $\epsilon^2$ &  $\tau_d/\tau_b$ &  &  & $\bol{v}_{\rm pol}/L$ & $\bol{v}_{\rm pol}$\\
 $\epsilon^3$ & & $E_{\parallel},P$ &  & $\tau_b^{-1},v_{\parallel}/L$ & $v_{\parallel}$ \\
  \hline
 \hline
\end{tabular}
\caption{\label{tab4} Drift wave turbulence ordering for the derivation of the GHM equation from a two-fluid model \cite{GHM,GHM2}. }
\end{center}
\begin{center}
\begin{tabular}{c c c c c c} 
 \hline
 \hline
 Order & Dimensionless & Fields & Distances & Rates & Velocities\\  
 \hline 
 $1$ & & $\bol{B},A_e$& $L$ & $\omega_c$ &\\ 
 $\epsilon$ & $\lambda\varphi,\omega_c^{-1}\p_t,L\nabla\log B,L\nabla \log A_e$ & $\bol{E}_{\perp}$ & & $\tau^{-1}_d,\bol{v}_{\bol{E}}/L$ & $\bol{v}_{\bol{E}}$\\
 $\epsilon^2$ & $\tau_d/\tau_b$  &  &  & $\bol{v}_{\rm pol}/L$ & $\bol{v}_{\rm pol}$\\
 $\epsilon^3$ & & $E_{\parallel},P$ &  & $\tau_b^{-1},v_{\parallel}/L$ & $v_{\parallel}$ \\
  \hline
 \hline
\end{tabular}
\caption{\label{tab5} Drift wave turbulence ordering for the derivation of the HM equation from a two-fluid model \cite{HM}.} 
\end{center}
\end{table}

The GHM equation \eqref{GHMgc5} inherits these features while extending 
the range of applicability of 
drift wave theory to inhomogeneous plasmas and general magnetic fields within a single partial differential equation. 
Therefore, it is expected to be useful to characterize 
drift wave turbulence in systems with strong density gradients and 
enhanced field curvature and inhomogeneity. 
In fact, drift wave type turbulence as well as so called entropy modes have been reported in experiments involving plasma confinement in dipole magnetic fields \cite{Boxer,Kenmochi,Garnier}. These systems typically exhibit a hot electron component, while the ion plasma is cold,  
suggesting the onset of drift wave dynamics. 
One of the possible applications of the GHM equation 
is therefore the study of electrostatic turbulence in planetary magnetospheres. 
Due to the guiding center origin of the GHM equation, 
the obtained results could be compared with the 
parent model represented by nonlinear gyrokinetic theory
\cite{Hahm96a,Hahm96b,Hahm09}.

Our purpose in this paper is to complement the theory 
pertaining to the GHM equation developed in \cite{GHM, GHM2}. 
In particular, we wish to elucidate the Hamiltonian structure \cite{PJM} of the GHM equation, 
and use it to infer the stability properties of steady solutions. 
Furthermore, we want to determine whether zonal flows can form in dipole magnetic fields, 
and characterize drift waves in dipole geometry.

The present paper is organized as follows.
In section 2 we examine the algebraic structure of the GHM equation, and obtain sufficient conditions on magnetic field $\bol{B}$ and electron spatial density $A_e$ under which the GHM equation defines a noncanonical Hamiltonian system. 
These results are consistent with the Hamiltonian structure of the standard HM equation 
\cite{Weinstein,Tassi,Hazeltine,Hazeltine2}.  
In section 3 we prove a theorem concerning the nonlinear stability of steady solutions of the GHM equation by applying the energy-Casimir method \cite{Holm,Tronci,Rein}. This result generalizes Arnold's stability criterion for a 2-dimensional fluid flow \cite{Arnold}. 
In section 4 we show that stable toroidal zonal flows can form in dipole magnetic fields, 
and characterize the angular frequency of drift waves in dipole geometry.
Concluding remarks are given in section 5.

\section{Algebraic structure of the GHM equation}

In this section, we discuss the algebraic structure of the GHM equation \eqref{GHMgc5}. In particular, we are concerned with the conditions under which equation \eqref{GHMgc5} can be written in the form
\begin{equation}
\frac{\p\eta}{\p t}=\lrc{\eta,H_{\Omega}},\label{Ham}
\end{equation}
where $\eta=\lambda A_e\chi-\sigma\omega$ and 
$\lrc{\cdot,\cdot}$ denotes a Poisson bracket \cite{PJM} acting on functionals of $\eta$.  

Conservation of energy $H_{\Omega}$ suggests that the GHM equation has an antisymmetric bracket structure. 
This antisymmetric bracket structure is sufficient to 
carry out the nonlinear stability analysis of the next section. 
However, we also expect  
the validity of the Jacobi identity (which would 
make the bracket also the Poisson bracket of a noncanonical Hamiltonian system) to depend on the geometry of the magnetic field. 
This expectation is made in analogy with the behavior of $\bol{E}\times\bol{B}$ drift dynamics, i.e. the dynamical system defined by
\begin{equation}
\dot{\bol{X}}=\bol{v}_{\bol{E}}=\frac{\bol{B}\cp\nabla\varphi}{B^2}.
\end{equation}
This dynamical system defines an Hamiltonian system only when the magnetic field has a vanishing helicity density, $\bol{B}\cdot\nabla\cp\bol{B}=0$, although the potential energy $Ze\varphi$, which represents the energy of the system, is a constant of motion for any $\bol{B}$ because  $\dot{\bol{X}}\cdot\nabla\varphi=0$ \cite{SatoPRE,Chandre}. 
Since the building block of drift wave turbulence is $\bol{E}\cp\bol{B}$ dynamics, 
we thus expect an integrability condition of the type $\bol{B}\cdot\nabla\cp\bol{B}=0$ 
to be required for the GHM equation to possess an Hamiltonian structure (recall that when $\bol{B}\cdot\nabla\cp\bol{B}=0$ there exist local functions
$\lambda,C$ such that $\bol{B}=\lambda\nabla C$ (Frobenius theorem \cite{Frankel})  implying that $C$ is a first integral, $\dot{\bol{X}}\cdot\nabla C=0$).   

First, define the second order linear partial differential operator $\mc{D}$ according to
\begin{equation}
\mc{D}\chi=\eta=\lambda A_e\chi-\sigma\omega=\lambda A_e\chi-\sigma\nabla\cdot\lr{A_e\frac{\nabla_{\perp}\chi}{B^2}}.
\end{equation}
In the following, we shall 
assume the inverse operator $\mc{D}^{-1}$ mapping $\eta$ to $\chi\in\mf{X}$ to be well defined by appropriate choice of the space of solutions $\mf{X}$.  
Next, consider the bracket
\begin{equation}
\lrc{F,G}=\int_{\Omega}A_e\lr{1-\sigma\frac{\bol{B}\cdot\nabla\cp\bol{v}_{\bol{E}}^{\chi}}{B^2}}\nabla\lr{\frac{\delta F}{\delta\eta}}\cdot\frac{\bol{B}}{B^2}\cp\nabla\lr{\frac{\delta G}{\delta\eta}}\,d\bol{x},\label{PB}
\end{equation}
acting on functionals $F,G\in\mf{X}^{\ast}$, where $\mf{X}^{\ast}$ denotes the dual space of $\mf{X}$. 
Assuming variations $\delta\chi$ and the electron spatial density $A_e$ to vanish on the boundary, and noting that
\begin{equation}
\frac{\delta H_{\Omega}}{\delta\eta}=\int_{\Omega}\frac{\delta H_{\Omega}}{\delta\chi\lr{\bol{x}',t}}\frac{\delta}{\delta\eta\lr{\bol{x},t}}\mc{D}^{-1}\eta\lr{\bol{x}',t}\,d\bol{x}'=\mc{D}^{-1}\frac{\delta H_{\Omega}}{\delta\chi}=\chi,
\end{equation}
where $H_{\Omega}$ is the energy given in \eqref{HM}, 
one can verify that the GHM equation \eqref{GHMgc5} can be written in the form \eqref{Ham} through the bracket \eqref{PB}. 

It is also clear that the bracket \eqref{PB} possesses an antisymmetric bracket structure. 
Indeed, the bracket \eqref{PB} is bilinear and alternating (and thus  antisymmetric), and it also satisfies the Leibniz rule. In formulae, 
\begin{subequations}
\begin{align}
&\lrc{aF+bG,H}=a\lrc{F,H}+b\lrc{G,H},~~~~\lrc{H,aF+bG}=a\lrc{H,F}+b\lrc{H,G},\\
&\lrc{F,F}=0,\\
&\lrc{F,G}=-\lrc{G,F},\\
&\lrc{FG,H}=F\lrc{G,H}+\lrc{F,H}G,
\end{align}\label{PBA}
\end{subequations}
for all $a,b\in\mathbb{R}$ and $F,G,H\in\mf{X}^{\ast}$.
For $\lrc{\cdot,\cdot}$ to qualify as a Poisson bracket it therefore remains to verify the Jacobi identity,
\begin{equation}
\lrc{F,\lrc{G,H}}+\lrc{G,\lrc{H,F}}+\lrc{H,\lrc{F,G}}=0.
\end{equation}
To this end, it is useful to introduce the following notation for the Jacobi identity,
\begin{equation}
\lrc{F,\lrc{G,H}}+\circlearrowright=0,
\end{equation}
where $\circlearrowright$ represents summation of even permutations.
Furthermore, we shall denote functional derivatives as $F_{\eta}=\delta F/\delta\eta$, and define the quantity
\begin{equation}
\bol{\beta}=A_e\lr{1-\sigma\frac{\bol{B}\cdot\nabla\cp\bol{v}_{\bol{E}}^{\chi}}{B^2}}\frac{\bol{B}}{B^2}.
\end{equation}
Notice that $\bol{\beta}=\bol{\beta}\lrs{\eta}$ is a functional of $\eta$. 
Omitting the range of integration, the Jacobi identity for the bracket \eqref{PB} now reads
\begin{equation}
\lrc{F,\lrc{G,H}}+\circlearrowright=\int\nabla F_{\eta}\cdot\bol{\beta}\cp\nabla \frac{\delta}{\delta\eta}\lr{\int\nabla G_{\eta}\cdot\bol{\beta}\cp\nabla H_{\eta}\,d\bol{x}}\,d\bol{x}+\circlearrowright.
\end{equation}
Terms involving second order functional derivatives of
$F$, $G$, and $H$ vanish (on this point see e.g. \cite{Olver}).
For example, the term
\begin{equation}
\int\nabla \lr{G_{\eta\eta}\delta\eta}\cdot\bol{\beta}\times\nabla H_{\eta}\,d\bol{x}=-\int \delta\eta G_{\eta\eta}\nabla H_{\eta}\cdot\nabla\cp\bol{\beta}\,d\bol{x}, 
\end{equation}
gives rise to the following contribution to the Jacobi identity,
\begin{equation}
-\int \nabla F_{\eta}\cdot\bol{\beta}\cp\nabla\lr{G_{\eta\eta}\nabla H_{\eta}\cdot\nabla\cp\bol{\beta}}\,d\bol{x}=-\int G_{\eta\eta}\nabla H_{\eta}\cdot\nabla\cp\bol{\beta}\nabla F_{\eta}\cdot\nabla\cp\bol{\beta}\,d\bol{x},\label{JI1}
\end{equation}
where we used the hypothesis that the electron density vanishes on the boundary, $A_e=0$ on $\p\Omega$, so that
$\bol{\beta}=\bol{0}$ on $\p\Omega$ and boundary terms evaluate to zero. 
On the other hand, the following term occurring in the  permutation $\lrc{H,\lrc{F,G}}$,
\begin{equation}
\int\nabla F_{\eta}\cdot\bol{\beta}\times\nabla \lr{G_{\eta\eta}\delta\eta}\,d\bol{x}=\int \delta\eta G_{\eta\eta}\nabla F_{\eta}\cdot\nabla\cp\bol{\beta}\,d\bol{x}, 
\end{equation}
contributes to the Jacobi identity with
\begin{equation}
\int\nabla H_{\eta}\cdot\bol{\beta}\cp\nabla\lr{G_{\eta\eta}\nabla F_{\eta}\cdot\nabla\cp\bol{\beta}}\,d\bol{x}=\int G_{\eta\eta}\nabla H_{\eta}\cdot\nabla\cp\bol{\beta}\nabla F_{\eta}\cdot\nabla\cp\bol{\beta}\,d\bol{x},
\end{equation}
which cancels with \eqref{JI1}.
It follows that the only surviving terms in the Jacobi identity are those involving functional derivatives of $\bol{\beta}$.
In particular, we must evaluate the integral
\begin{equation}
\int \nabla G_{\eta}\cdot\delta\bol{\beta}\cp\nabla H_{\eta}\,d\bol{x}.
\end{equation}
To this end, it is useful to define the quantities
\begin{equation}
\zeta={1-\sigma\frac{\bol{B}\cdot\nabla\cp\bol{v}_{\bol{E}}^{\chi}}{B^2}},~~~~
\theta=-\sigma\nabla G_{\eta}\cdot \frac{\bol{B}}{B^2}\cp\nabla H_{\eta},
\end{equation}
so that
\begin{equation}
\begin{split}
\int \nabla G_{\eta}\cdot\delta\bol{\beta}\cp\nabla H_{\eta}\,d\bol{x}=&\int A_e\theta{\frac{\bol{B}}{B^2}\cdot\nabla\cp\lr{\frac{\bol{B}\cp\nabla\delta\chi}{B^2}}}\,d\bol{x}\\
=&\int \frac{\bol{B}\cp\nabla\delta\chi}{B^2}\cdot\nabla\cp\lr{A_e\theta\frac{\bol{B}}{B^2}}\,d\bol{x}\\=&\int \nabla \mc{D}^{-1}\delta\eta\cdot\nabla\cp\lr{A_e\theta\frac{\bol{B}}{B^2}}\cp\frac{\bol{B}}{{B^2}}\,d\bol{x}\\=&
\int \delta\eta\mc{D}^{-1}\nabla\cdot\lrc{\frac{\bol{B}}{B^2}\cp\lrs{\nabla\cp\lr{
\theta A_e\frac{\bol{B}}{B^2}}
}}\,d\bol{x}
\end{split}
\end{equation}
and the Jacobi identity can be written as
\begin{equation}
\begin{split}
\lrc{F,\lrc{G,H}}+\circlearrowright=&\int A_e\zeta\nabla F_{\eta}\cdot\frac{\bol{B}}{B^2}\cp\nabla\mc{D}^{-1}\nabla\cdot\lrc{\frac{\bol{B}}{B^2}\cp\lrs{\nabla\cp\lr{
\theta A_e\frac{\bol{B}}{B^2}}
}}\,d\bol{x}+\circlearrowright\\
=&\int
\mc{D}^{-1}\nabla\cdot\lrc{\frac{\bol{B}}{B^2}\cp\lrs{\nabla\cp\lr{
\theta A_e\frac{\bol{B}}{B^2}}
}}
\nabla F_{\eta}\cdot\nabla\cp\lr{\zeta A_e\frac{\bol{B}}{B^2}}\,d\bol{x}+\circlearrowright
.\label{JI2}
\end{split}
\end{equation}
Since the value of the parameters $\sigma$ and $\lambda$ is not specified, 
terms proportional to different powers of $\sigma$ must cancel separately. The Jacobi identity above contains terms scaling as $\sigma\mc{D}^{-1}$, terms scaling as $\sigma^2\mc{D}^{-1}$, and terms scaling as $\sigma^3\mc{D}^{-1}$.   
From the first group of terms, we obtain the condition 
\begin{equation}
\begin{split}
\int&\mc{D}^{-1}\nabla\cdot\lrc{\frac{\bol{B}}{B^2}\cp\lrs{\nabla\cp\lr{
\theta A_e\frac{\bol{B}}{B^2}}
}}\nabla F_{\eta}\cdot\nabla\cp\lr{A_e\frac{\bol{B}}{B^2}}\,d\bol{x}+\circlearrowright
\\
&=\int\lrc{\mc{D}^{-1}\nabla\cdot\lrc{\frac{\bol{B}}{B^2}\cp\lrs{\nabla\cp\lr{
\theta A_e\frac{\bol{B}}{B^2}}
}}\nabla F_{\eta}+\circlearrowright}\cdot\nabla\cp\lr{A_e\frac{\bol{B}}{B^2}}\,d\bol{x}=0.\label{JI3}
\end{split}
\end{equation}
We therefore see that a sufficient condition for this quantity to vanish is that the magnetic field $\bol{B}$ and the spatial density $A_e$ satisfy
\begin{equation}
\nabla\cp\lr{A_e\frac{\bol{B}}{B^2}}=\bol{0}.\label{JIC}
\end{equation}
Now observe that when \eqref{JIC} holds, 
the surviving terms in the Jacobi identity \eqref{JI2} are
\begin{equation}
\begin{split}
\lrc{F,\lrc{G,H}}+\circlearrowright
=&\int\mc{D}^{-1}\nabla\cdot\lr{A_e\frac{\nabla_{\perp}\theta}{B^2}}\nabla F_{\eta}\cdot\nabla\zeta\cp{ A_e\frac{\bol{B}}{B^2}}\,d\bol{x}+\circlearrowright\\
=&\frac{1}{\sigma}\int \lrs{-{\theta}+{\lambda}\mc{D}^{-1}\lr{A_e\theta}}\nabla F_{\eta}\cdot\nabla\zeta\cp A_e\frac{\bol{B}}{B^2}\,d\bol{x}+\circlearrowright
.\label{JI4}
\end{split}
\end{equation}
On the other hand, the condition \eqref{JIC} implies that
there exists some local function $C$ such that $A_e\bol{B}/B^2=\nabla C$ (Poincar\'e lemma). 
The first term within the integrand involving $-\theta$ can therefore be locally written as
\begin{equation}
\begin{split}
&A_e^{-1}\nabla G_{\eta}\cdot\nabla C\times\nabla H_{\eta}\nabla F_{\eta}\cdot\nabla \zeta\cp\nabla C+\circlearrowright=
A_e^{-1}\\
&\lrs{\frac{\p G_{\eta}}{\p x}\lr{\frac{\p C}{\p y}\frac{\p H_{\eta}}{\p z}-\frac{\p C}{\p z}\frac{\p H_{\eta}}{\p y}}
+\frac{\p G_{\eta}}{\p y}\lr{\frac{\p C}{\p z}\frac{\p H_{\eta}}{\p x}-\frac{\p C}{\p x}\frac{\p H_{\eta}}{\p z}}
+\frac{\p G_{\eta}}{\p z}\lr{\frac{\p C}{\p x}\frac{\p H_{\eta}}{\p y}-\frac{\p C}{\p y}\frac{\p H_{\eta}}{\p x}}}\\
&\lrs{\frac{\p F_{\eta}}{\p x}\lr{\frac{\p {\zeta}}{\p y}\frac{\p C}{\p z}-\frac{\p {\zeta}}{\p z}\frac{\p C}{\p y}}
+\frac{\p F_{\eta}}{\p y}\lr{\frac{\p \zeta}{\p z}\frac{\p C}{\p x}-\frac{\p \zeta}{\p x}\frac{\p C}{\p z}}
+\frac{\p F_{\eta}}{\p z}\lr{\frac{\p \zeta}{\p x}\frac{\p {C}}{\p y}-\frac{\p \zeta}{\p y}\frac{\p {C}}{\p x}}}+\circlearrowright=0.\label{vanish}
\end{split}
\end{equation}
Unfortunately, the term in \eqref{JI4} containing $\mc{D}^{-1}\lr{A_e\theta}$ 
appears to represent an obstruction to the Jacobi identity that cannot be trivially removed.
This fact suggests that in order to fulfill the Jacobi identity when 
condition \eqref{JIC} holds, the bracket \eqref{PB} itself must be modified. 
To this end, define the following alternative bracket:
\begin{equation}
\lrc{F,G}'=\int_{\Omega}\eta\nabla\lr{\frac{\delta F}{\delta\eta}}\cdot\frac{\bol{B}}{B^2}\cp\nabla\lr{\frac{\delta G}{\delta\eta}}\,d\bol{x}.\label{PB2}
\end{equation}
Observe that the bracket \eqref{PB2} satisfies the antisymmetric bracket axioms \eqref{PBA} by the same arguments used for the bracket \eqref{PB}. In addition, if \eqref{JIC} holds, the GHM equation \eqref{GHMgc5} can be written in the equivalent form
\begin{equation}
\frac{\p\eta}{\p t}=\lrc{\eta,H_{\Omega}}'.
\end{equation}
Furthermore, by repeating the same steps as above the Jacobi identity for the new bracket \eqref{PB2} can be evaluated to be
\begin{equation}
\begin{split}
\lrc{F,\lrc{G,H}'}'+\circlearrowright=\int A_e \nabla\lr{\frac{\eta}{A_e}}\cdot\frac{\bol{B}}{B^2}\cp\nabla F_{\eta}\, \nabla G_{\eta}\cdot\frac{\bol{B}}{B^2}\cp\nabla H_{\eta}\,d\bol{x}+\circlearrowright=0, 
\end{split}
\end{equation}
which vanishes by the same calculation used in equation \eqref{vanish}. 
We have thus shown that the antisymmetric bracket \eqref{PB2} 
is a Poisson bracket whenever equation \eqref{JIC} holds. 
It should not be surprising that \eqref{JIC} is exactly the same condition 
for the conservation of generalized entrophy $W_{\Omega}$ (see table 1).
Indeed, the functional $W_{\Omega}$ 
is a Casimir invariant of the Poisson bracket \eqref{PB2},
\begin{equation}
\begin{split}
\frac{dW_{\Omega}}{dt}=&\lrc{W_{\Omega},H_{\Omega}}'\\=&\int{\frac{\eta}{A_e}}\nabla w'\cdot A_e\frac{\bol{B}}{B^2}\cp\nabla \frac{\delta H_{\Omega}}{\delta\eta}\,d\bol{x}\\=&\int_{\p\Omega}
A_e\lrs{\int{\frac{\eta}{A_e}}w''d\lr{\frac{\eta}{A_e}} }
\frac{\bol{B}}{B^2}\cp\nabla\frac{\delta H_{\Omega}}{\delta\eta}\cdot\bol{n}\,dS=0~~~~\forall H_{\Omega}.
\end{split}
\end{equation}
In the last passage, we used the boundary condition $A_e=0$ on $\p\Omega$. 
We stress again that, however, the  bracket \eqref{PB2} cannot be used to generate the GHM system \eqref{GHMgc5} when the condition \eqref{JIC} does not hold. 
We also remark that the mass $M_{\Omega}$ encountered in equation \eqref{MH} is a Casimir invariant of both brackets, i.e.
\begin{equation}
\frac{dM_{\Omega}}{dt}=\lrc{M_{\Omega},H_{\Omega}}=\lrc{M_{\Omega},H_{\Omega}}'=0~~~~\forall H_{\Omega}. 
\end{equation}
In this calculations we used the fact that the boundary condition $A_e=0$ on $\p\Omega$ implies that
\begin{equation}
{\delta M_{\Omega}}=\int_{\Omega}\lrs{\lambda A_e\delta\chi-\nabla\cdot\lr{A_e\frac{\nabla_{\perp}\delta\chi}{B^2}}}\,d\bol{x}=\int_{\Omega}\delta\eta\,d\bol{x}.
\end{equation}

It is worth observing that the condition \eqref{JIC} implies that the magnetic field satisfies 
the Frobenius integrability condition $\bol{B}\cdot\nabla\cp\bol{B}=0$ because $\bol{B}=A_e^{-1}B^2\nabla C$ locally.   
Furthermore, it also implies that the $\bol{E}\cp\bol{B}$ velocity $\bol{v}_{\bol{E}}^{\chi}$ multiplied by the spatial density $A_e$ is divergence free, $\nabla\cdot \lr{A_e\bol{v}_{\bol{E}}^{\chi}}=\nabla\chi\cdot\nabla\cp\lr{A_eB^{-2}\bol{B}}=0$.
Notice also that \eqref{JIC} can always be satisfied for a vacuum field $\bol{B}=\nabla C$ by setting $A_e\propto B^2$. 
Finally, when $A_e$ is a constant the GHM equation \eqref{GHMgc5} defines a noncanonical Hamiltonian system provided that the magnetic field $\bol{B}$ satisfies
\begin{equation}
\nabla\cdot\bol{B}=0,~~~~\nabla\cp\lr{\frac{\bol{B}}{B^2}}=\bol{0}.
\end{equation}
Nontrivial examples of such configurations in different geometries can be found in \cite{GHM}.


\section{Nonlinear stability}
The aim of this section is to elucidate the nonlinear stability properties of steady solutions of the GHM equation \eqref{GHMgc5} with the aid of the energy-Casimir method \cite{Holm,Tronci,Rein,Arnold}.

First, notice that steady solutions $\chi_0\lr{\bol{x}}$ of the GHM equation \eqref{GHMgc5}
can be characterized in terms of critical points of the energy-Casimir functional 
\begin{equation}
\mf{H}_{\Omega}=H_{\Omega}+\gamma M_{\Omega}+\nu W_{\Omega},\label{mfH}
\end{equation}
where $\gamma,\nu$ are spatial constants, and $\nu$ is taken to be zero for configurations violating the Poisson bracket condition \eqref{JIC}. 
Indeed, when $\delta \mf{H}_{\Omega}=0$, 
from \eqref{Ham} one sees that  $\eta_t=0$.
Let $\chi\lr{\bol{x},t}$ denote a solution of the GHM equation \eqref{GHMgc5}. 
A critical point $\chi_0$ is nonlinearly stable provided 
that for every $\epsilon>0$ there exists a norm $\Norm{\cdot}_{1}$ on the space of solutions $\mf{X}$ and a $\delta>0$ such that  $\Norm{\chi\lr{\bol{x},0}-\chi_0\lr{\bol{x}}}_{1}<\delta$ implies  
\begin{equation}
\Norm{\chi\lr{\bol{x},t}-\chi_0\lr{\bol{x}}}_2< \epsilon~~~~\forall t\geq0,\label{NS}
\end{equation}
where $\Norm{\cdot}_2$ is 
a further norm on the state space $\mf{X}$. 
Notice that the nonlinear stability described by \eqref{NS}
only ensures that the solution $\chi$ remains close to the critical point in the norm $\Norm{\cdot}_2$.

\begin{theorem}{(Nonlinear stability of steady solutions of the GHM equation)}
Let $\chi_0\lr{\bol{x}}\in C^{2}\lr{{\Omega}}$ denote a critical point of the energy-Casimir functional  $\mf{H}_{\Omega}$. 
If the condition ${\nabla\cp\lr{{A_e}\bol{B}/B^2}}=\bol{0}$ of equation \eqref{JIC} holds, 
assume that the function $w\lr{\eta/A_e}$ appearing within the integrand of the Casimir invariant $W_{\Omega}$ 
is twice differentiable in its argument, and that it 
satisfies 
\begin{equation}
0 < c_{m}\leq \nu w''=\nu \frac{d^2 w}{d\lr{\eta/A_e}^2}\leq c_M<\infty,\label{w2}
\end{equation}
with $c_m$ and $c_M$ real constants. 
If ${\nabla\cp\lr{{A_e}\bol{B}/B^2}}\neq\bol{0}$ set $\nu=0$. 
Further assume that $\bol{B},A_e\in C^{2}\lr{\bar{\Omega}}$, that  their 
minima satisfy $B_{m},A_{em}>0$, and that the GHM equation \eqref{GHMgc5} admits a solution  $\chi\lr{\bol{x},t}\in C^2\lr{{\Omega}\cp[0,t)}$ for all ${t}\geq 0$ such that $\delta\chi=\chi-\chi_0=0$ and $A_e={0}$ on the boundary  $\p\Omega$. 
Then, the critical point $\chi_0$ is nonlinearly stable:  there exists a positive real constant $\mf{C}$ such that
\begin{equation}
\Norm{\chi\lr{t}-\chi_0}_{\perp}^2\leq \mf{C}\Norm{\chi\lr{0}-\chi_0}_{\perp}^2~~~~\forall t\geq 0,
\end{equation}
with
\begin{equation}
\Norm{\chi}^2_{\perp}=
\begin{cases}
\Norm{\chi}^2_{L^2\lr{\Omega}}+
\Norm{\nabla_{\perp}\chi}^2_{L^2\lr{\Omega}}+
\Norm{\mc{D}\chi}^2_{L^2\lr{\Omega}},~~~~{\rm if}~~\nabla\cp\lr{A_e\frac{\bol{B}}{B^2}}=\bol{0},\\
\Norm{\chi}^2_{L^2\lr{\Omega}}+
\Norm{\nabla_{\perp}\chi}^2_{L^2\lr{\Omega}}~~~~{\rm if}~~\nabla\cp\lr{A_e\frac{\bol{B}}{B^2}}\neq\bol{0}, 
\end{cases}
\end{equation}
where $L^2\lr{\Omega}$ denotes the standard $L^2$ norm in $\Omega$ and we used the abbreviated notation $\chi\lr{t}=\chi\lr{\bol{x},t}$. 
\end{theorem}

\begin{proof}
We start by observing that key to the proof is the conservation of $\mf{H}_\Omega$.  
Indeed, the energy-Casimir method consists in finding norms $\Norm{\cdot}_1$ 
and $\Norm{\cdot}_2$ on $\mf{X}$ so that the following chain of inequalities holds:
\begin{equation}
\mc{C}\Norm{\chi\lr{0}-\chi_0}_1^2\geq\abs{\mf{H}_{\Omega}\lrs{\chi\lr{0}}-\mf{H}_{\Omega}\lrs{\chi_0}}=\abs{\mf{H}_{\Omega}\lrs{\chi\lr{t}}-\mf{H}_{\Omega}\lrs{\chi_0}}\geq\mc{C}'\Norm{\chi\lr{t}-\chi_0}_2^2,\label{chain}
\end{equation}
where $\mc{C},\mc{C}'$ are positive real constants.  
To derive these inequalities 
for the case $\nu\neq 0$ (corresponding to $\nabla\cp\lr{A_e\bol{B}/B^2}=\bol{0}$)  
we rely on a standard result:
setting $\eta=\mc{D}\chi$,   $\eta_0=\mc{D}\chi_0$, and  $\delta\eta=\eta-\eta_0$, 
Taylor's theorem asserts that  
\begin{equation}
w\lr{\frac{\eta}{A_e}}=w\lr{\frac{\eta_0}{A_e}}+w'\lr{\frac{\eta_0}{A_e}}\frac{\delta\eta}{A_e}+w''\lr{\frac{\tilde{\eta}}{A_e}}\frac{\delta\eta^2}{2A_e^2},
\end{equation}
with $\tilde{\eta}$ between $\eta_0$ and $\eta$ and $w'=dw/d\lr{\eta/A_e}$. 
Since $A_e\in C^2\lr{\bar{\Omega}}$ and $A_e\geq A_{em}>0$, $A_e$ attains a positive maximum $A_{eM}<\infty$ in $\bar{\Omega}$. 
Using $0<c_m\leq\nu w''\leq c_M<\infty$ and $0<A_{em}\leq A_e\leq A_{eM}<\infty$ we thus obtain 
\begin{equation}
\frac{c_m}{2A_{eM}}\Norm{\delta\eta}^2_{L^2\lr{\Omega}}\leq\int_{\Omega}\frac{\nu\delta\eta^2}{2A_e}w''\lr{\frac{\tilde{\eta}}{A_e}}\,d\bol{x}\leq \frac{c_M}{2A_{em}}\Norm{\delta\eta}^2_{L^2\lr{\Omega}}.\label{Taylor}
\end{equation}
Now observe that
\begin{equation}
\begin{split}
\mf{H}_{\Omega}\lrs{\chi\lr{t}}-\mf{H}_{\Omega}\lrs{\chi_0}=&\int_{\Omega}A_e\lrc{\frac{\lambda}{2}\lr{2\chi_0\delta\chi+\delta\chi^2}
+\sigma\frac{2\nabla_{\perp}\delta\chi\cdot\nabla_{\perp}\chi_0+\abs{\nabla_{\perp}\delta\chi}^2}{2B^2}
+\gamma\lambda\delta\chi}\,d\bol{x}\\
&+\nu\int_{\Omega}A_e\lrs{w'\lr{\frac{\eta_0}{A_e}}\frac{\delta\eta}{A_e}+w''\lr{\frac{\tilde{\eta}}{A_e}}\frac{\delta\eta^2}{2A_e^2}}\,d\bol{x}\\
=&
\int_{\Omega}A_e\lrc{\frac{\lambda}{2}\lr{2\chi_0\delta\chi+\delta\chi^2}
-\frac{\sigma}{A_e} \delta\chi\nabla\cdot\lr{A_e\frac{\nabla_{\perp}\chi_0}{B^2}}+\sigma\frac{\abs{\nabla_{\perp}\delta\chi}^2}{2B^2}
+\gamma\lambda\delta\chi}\,d\bol{x}\\
&+\nu\int_{\Omega}A_e\lrs{w'\lr{\frac{\eta_0}{A_e}}\frac{\delta\eta}{A_e}+w''\lr{\frac{\tilde{\eta}}{A_e}}\frac{\delta\eta^2}{2A_e^2}}\,d\bol{x}
.
\end{split}\label{dHt}
\end{equation}
However, by hypothesis $\chi_0$ solves the critical   
equation for $\mf{H}_{\Omega}$
\begin{equation}
\lambda A_e\lr{ \chi_0+\gamma}-\sigma\nabla\cdot\lr{A_e\frac{\nabla_{\perp}\chi_0}{B^2}}
+\nu\mc{D}w'=\mc{D}\lr{\chi_0+\gamma
+\nu w'}=0.\label{eq}
\end{equation}
Hence, the difference \eqref{dHt} reduces to
\begin{equation}
\begin{split}
\mf{H}_{\Omega}\lrs{\chi\lr{t}}-\mf{H}_{\Omega}\lrs{\chi_0}=&
\int_{\Omega}A_e\lrc{\frac{\lambda}{2}\delta\chi^2
+\sigma\frac{\abs{\nabla_{\perp}\delta\chi}^2}{2B^2}
+{\frac{\nu\delta\eta^2}{2A_e^2}w''\lr{\frac{\tilde{\eta}}{A_e}}}}\,d\bol{x}.
\label{hth0}
\end{split}
\end{equation}
Using \eqref{Taylor}, it readily follows that
\begin{equation}
\begin{split}
\frac{1}{2}&\lr{\lambda A_{em}\Norm{\delta\chi}^2_{L^2\lr{\Omega}}+\frac{\sigma A_{em}}{B_{M}^2}\Norm{\nabla_{\perp}\delta\chi}^2_{L^2\lr{\Omega}}+\frac{c_{m}}{A_{eM}}\Norm{\mc{D}\delta\chi}^2_{L^2\lr{\Omega}}}\leq\mf{H}_{\Omega}\lrs{\chi\lr{t}}-\mf{H}_{\Omega}\lrs{\chi_0}\\&=\mf{H}_{\Omega}\lrs{\chi\lr{0}}-\mf{H}_{\Omega}\lrs{\chi_0}\leq \frac{1}{2}\lr{\lambda A_{eM}\Norm{\delta\chi_0}^2_{L^2\lr{\Omega}}+\frac{\sigma A_{eM}}{B_{m}^2}\Norm{\nabla_{\perp}\delta\chi_0}^2_{L^2\lr{\Omega}}+\frac{c_{M}}{A_{em}}\Norm{\mc{D}\delta\chi_0}^2_{L^2\lr{\Omega}}},
\end{split}
\end{equation}
where $\delta\chi_0=\chi\lr{0}-\chi_0$ and $B_M<\infty$ is the maximum of $B$. We have thus shown that 
\begin{equation}
\Norm{\chi\lr{t}-\chi_0}_{\perp}^2\leq \mf{C}\Norm{\chi\lr{0}-\chi_0}_{\perp}^2~~~~\forall t\geq 0,
\end{equation}
for some positive real constant $\mf{C}$. 
The case $\nu=0$ (corresponding to $\nabla\cp\lr{A_e\bol{B}/B^2}\neq\bol{0}$) follows in a similar fashion 
and the theorem is proven. 
\end{proof}

The following remarks are useful.

\begin{remark}
Theorem 1 generalizes Arnold's result concerning the stability of a two dimensional ideal fluid flow \cite{Arnold}. 
Indeed, Arnold's case can be recovered by setting $\bol{B}=\nabla z$, 
$A_e=\sigma=1$, and $\lambda=0$. In this setting we have 
\begin{equation}
\eta=-\Delta_{\lr{x,y}}\chi,
\end{equation}
so that the critical point equation \eqref{eq} reduces to
\begin{equation}
{\chi_0+\nu w'\lr{-\Delta_{\lr{x,y}}\chi_0}}=0.
\end{equation}
Hence, using Arnold's notation, 
\begin{equation}
\nu w''=\frac{\nabla\chi_0}{\nabla\Delta_{\lr{x,y}}\chi_0}=-\frac{\nabla\chi_0}{\nabla\eta_0}.
\end{equation}
\end{remark}

\begin{remark}
According to theorem 1 steady states 
of the GHM equation \eqref{GHMgc5} 
corresponding to $\nu=0$ are nonlinearly stable, provided that the hypothesis of theorem 1  pertaining to regularity and boundary conditions hold true. 
Notice also that $\nu=0$ when the magnetic field $\bol{B}$ and the electron spatial density $A_e$ do not satisfy the condition \eqref{JIC} and the generalized enstrophy $W_{\Omega}$ is not a constant of motion.
\end{remark}

\section{Zonal flows and drift waves in dipole magnetic fields}
As outlined in the introduction, one of the motivations  
behind the development of GHM equation \eqref{GHMgc5} is the understanding of    
drift wave turbulence in complex magnetic geometries, 
such as that of a magnetospheric plasma.
The purpose of this last section is to show that the theory developed 
in this paper points to the existence of stable toroidal zonal flows 
with radial velocity shear within dipole magnetic fields, and to characeterize drift waves in dipole geometry. 
To see this, we first observe that a dipole magnetic field is a vacuum field outside the central region containing the electric current generating it. 
Furthermore, it is axially symmetric. We may therefore write
\begin{equation}
\bol{B}=\nabla\zeta\lr{r,z}=\nabla\Psi\lr{r,z}\cp\nabla\phi,
\end{equation}
where $\lr{r,\phi,z}$ denote cylindrical coordinates, $\zeta\lr{r,z}$ the magnetic potential, and $\Psi\lr{r,z}$ the flux function. 
It is convenient to work with magnetic coordinates $\lr{\zeta,\Psi,\phi}$. The Jacobian determinant of this coordinate change is
\begin{equation}
\nabla\zeta\cdot\nabla\Psi\cp\nabla\phi=B^2.
\end{equation}
Then, it follows that
\begin{subequations}
\begin{align}
\bol{v}_{\bol{E}}^{\chi}=&\chi_{\Psi}\p_{\phi}-\chi_{\phi}\p_{\Psi},\\
A_e\frac{\bol{B}\cdot\nabla\cp\bol{v}_{\bol{E}}^{\chi}}{B^2}=&\frac{A_e}{B^2}\nabla\cdot\nabla_{\perp}\chi,
\end{align}
\end{subequations}
where we used the notation $\lr{\p_{\zeta},\p_{\Psi},\p_{\phi}}$ for tangent vectors. 
Hence, 
the GHM equation \eqref{GHMgc5} can be written as
\begin{equation}
\frac{\p}{\p t}\lrs{\lambda A_e\chi-\sigma\nabla\cdot\lr{A_e\frac{\nabla_{\perp}\chi}{B^2}}}=B^2\lrs{\chi,\frac{A_e}{B^2}\lr{\sigma\frac{\Delta_{\perp}\chi}{B^2}-1}}_{\lr{\Psi,\phi}},\label{GHMdipole}
\end{equation}
where we introduced the linear differential operators $\Delta_{\perp}=\nabla\cdot\nabla_{\perp}$ and $\lrs{f,g}_{\lr{\Psi,\phi}}=f_{\Psi}g_{\phi}-f_{\phi}g_{\Psi}$. 
Notice that this equation is two-dimensional, i.e. it 
can be considered as a closed system within a surface given by a level set of  $\zeta$, with the function $\zeta$ effectively behaving as an external parameter. 
It follows that steady solutions $\chi_0$ of equation \eqref{GHMdipole} 
satisfy
\begin{equation}
\frac{A_e}{B^2}\lr{\sigma\frac{\Delta_{\perp}\chi_0}{B^2}-1}=f\lr{\chi_0,\zeta},\label{steadydipole}
\end{equation}
with $f\lr{\chi_0,\zeta}$ some function of $\chi_0$ and $\zeta$. 
Steady solutions
with given values of mass and generalized enstrophy 
can be equivalently  characterized in terms of critical points of the energy-Casimir function, which, recalling \eqref{eq}, are given by 
\begin{equation}
\mc{D}\lr{\chi_0+\gamma+\nu w'}=0.\label{steadydipole2}
\end{equation}
In this context, a steady zonal flow solution is described by the condition $\chi_{0\phi}=0$, implying a toroidal flow $\bol{v}_{\bol{E}}^{\chi}=\chi_{0\Psi}\p_{\phi}$. 
Evidently, equations \eqref{steadydipole} and \eqref{steadydipole2} admit  such configurations 
provided that $A_{e}$ is axially symmetric (since the dipole magnetic field $\bol{B}$ is axially symmetric). 
Observe that in this case on the equatorial plane $z=0$ the toroidal $\bol{E}\cp\bol{B}$ velocity
has radial shear since
$\bol{v}_{\bol{E}}^{\chi}\lr{r,z=0}=\chi_{0\Psi}\lr{r,z=0}\p_{\phi}$. 
The stability properties of these zonal flow solutions 
can be deduced from theorem 1. In particular, 
they will depend on the specific value of the generalized vorticity $W_{\Omega}$ in the case in which the density $A_e$ satisfies \eqref{JIC}, i.e. $A_e\propto B^2$ (configurations of this type  
are predicted by 
equilibrium statistical mechanics 
because the invariant (Liouville) measure associated with $\bol{E}\cp\bol{B}$ dynamics in a vacuum field is $B^2\,d\bol{x}$ \cite{SY}). Otherwise $\nu=0$, and zonal flows are expected to be nonlinearly stable.
It should be emphasized that the characteristic spatial scale of $\chi_0$ is related to that of  magnetic field $\bol{B}$ and electron spatial density $A_e$, while the zonal nature of the solution stems from the axial symmetry of these fields.
Nevertheless, exception made for the 
case in which $A_e$ satisfies  \eqref{JIC}, the  generalized enstrophy $W_{\Omega}$ is not constant, and inverse energy cascade toward small wave numbers is not available in the usual form.  
The turbulent mechanism by which zonal flow solutions can be formed in a dipole field therefore requires a separate discussion. 
A crucial role should be played by boundary conditions for $\chi_0$, especially when $\nu=0$ and there is no constraint arising from generalized enstrophy, since trivial boundary conditions, such as Dirichlet boundary conditions, result in trivial steady states $\chi_0+\gamma=0$.

We conclude this section by describing the drift wave in a dipole magnetic field. 
Assume that the electron spatial density $A_e=A_e\lr{\zeta,\Psi}$ is axially symmetric. Let 
\begin{equation}
\chi_d=\xi\lr{\zeta,\Psi}\exp\lrc{-{\rm i}\lr{\ell\phi+\omega t}} 
\end{equation}
represent the drift wave with $\ell\in\mathbb{Z}$, $\xi\lr{\zeta,\Psi}$ a real function of $\zeta$ and $\Psi$, and $\omega\in\mathbb{R}$.  
Linearizing equation \eqref{GHMdipole} 
with respect to $\chi_d$ 
we thus obtain the following equation for $\xi$,
\begin{equation}
\frac{1}{A_e}\nabla\cdot\lr{A_e\frac{\nabla_{\perp}\xi}{B^2}}+\xi\lrs{\frac{\ell}{\sigma\omega}\frac{\p}{\p\Psi}\log\lr{\frac{A_e}{B^2}}-\frac{\ell^2}{r^2B^2}-\frac{\lambda}{\sigma}}=0.
\end{equation}
Conversely, the angular frequency $\omega$ can be expressed as
\begin{equation}
\omega=\frac{\ell\frac{\p}{\p\Psi}\log\lr{\frac{A_e}{B^2}}}{\frac{\sigma\ell^2}{r^2B^2}+\lambda-\frac{\sigma}{A_e\xi}\nabla\cdot\lr{A_e\frac{\nabla_{\perp}\xi}{B^2}}}.
\end{equation}
In order to estimate the magnitude of $\omega$, consider the simplified case in which $\ell$ is small and $\log\lr{A_e/B^2}$ is a weak function of $\Psi$, and consider its Taylor expansion around $\Psi_0$. Then, we may set $\xi=\xi\lr{\zeta}$ to find 
\begin{equation}
\omega\approx \frac{\ell\beta_{\Psi_0}}{\lambda}=\ell 
\beta_{\Psi_0}\,T_{e}\lrs{eV},~~~~\beta_{\Psi_0}=\lrs{\frac{\p}{\p\Psi}\log\lr{\frac{A_e}{B^2}}}_{\Psi=\Psi_0}.
\end{equation}
 Here, $T_e\lrs{eV}=\lambda^{-1}$ is the electron temperature expressed in electronvolt. 
 Notice that the term $\log A_e$ is the one responsible for the usual drift wave in the HM equation. 
 Remarkably, even in the presence of a constant electron spatial density $A_e$,  
 an inhomogeneous magnetic field can sustain a geometric drift wave through the spatial dependence of $B$. 
 For a dipole magnetic field $B\sim 1\,T$ in a trap with size $L\sim 1\,m$, a roughly constant electron spatial density $A_e$, $\ell=1$, and an electron temperature of $1\,keV$, one obtains $\omega\approx 10^3\,Hz$. These values are compatible with experimental measurements (see \cite{Kenmochi}).

Finally, we observe that the standard dispersion relation for the drift wave 
in a straight homogeneous magnetic field can be recovered by setting $\bol{B}=B_0\nabla z$, $\xi=\exp\lrc{{\rm i}k_xx}$, $\ell=-k_yL$, $\log A_e=\log A_{e0}+\beta x$, $\zeta=B_0z$, $\Psi=B_0L x$, and $\phi=y/L$ with $B_0,\omega,k_x,k_y,A_{e0},\beta\in\mathbb{R}$ and $L\beta\sim\epsilon<<1$ in equation \eqref{GHMdipole}. In this case, we have
\begin{equation}
\omega=-\frac{k_y\beta}{\lambda B_0+\sigma
\frac{k_x^2+k_y^2}{B_0}}.
\end{equation}

%





\section{Concluding remarks}
The generalized Hasegawa-Mima (GHM) equation \eqref{GHMgc5} is a nonlinear equation describing the evolution of electrostatic turbulence in inhomogeneous plasmas immersed in a static magnetic field with arbitrary geometry.
The GHM equation serves as a generalization of the standard Hasegawa-Mima (HM) equation 
for drift wave turbulence in a straight homogeneous magnetic field, 
and it can be applied to ion-electron plasmas  
characterized by strong inhomogeneities of both the magnetic field $\bol{B}$ and
the electron spatial density $A_e$.  
In particular, the equation can account for  
turbulence occurring over spatial scales comparable to 
the characteristic spatial scales of the background magnetic field,  
and it can be used to model electrostatic turbulence in 
systems with irregular geometries, such as 
the dipole magnetic field of a planetary magnetosphere or the confining magnetic field of a stellarator.

In this study, we examined the conditions under which the GHM equation possesses a noncanonical Hamiltonian structure. We found that the antisymmetric bracket \eqref{PB2} becomes a Poisson bracket whenever the magnetic field $\bol{B}$ 
and the electron spatial density $A_e$ fulfill the integrability condition \eqref{JIC}. 
This same condition is required for the conservation of generalized enstrophy $W_{\Omega}$, which is a Casimir invariant of the Poisson bracket \eqref{PB2}. 
Using the algebraic structure of the GHM equation, we applied 
the energy-Casimir method to obtain a nonlinear stability criterion for steady solutions of the GHM equation \eqref{GHMgc5} (theorem 1).
This result implies that sufficiently regular solutions of the GHM equation, whose initial conditions are sufficiently close to critical points of the energy-Casimir function \eqref{mfH} characterized either by $\nu=0$ or \eqref{w2}, remain close to these critical points at all later times. 
Finally, we showed that 
radially sheared stable toroidal zonal flows may be created in dipole magnetic fields,  
and characterized the angular frequency of magnetospheric drift waves, which explicitly  depends on the magnetic field geometry.

\section*{Statements and declarations}

\subsection*{Data availability}
Data sharing not applicable to this article as no datasets were generated or analysed during the current study.

\subsection*{Funding}
The research of NS was partially supported by JSPS KAKENHI Grant No. 21K13851 
and No.  22H04936.

\subsection*{Competing interests} 
The authors have no competing interests to declare that are relevant to the content of this article.



\end{document}